\documentclass[JHEP,12pt]{article}
\usepackage{jheppub}
\usepackage{amssymb,amsmath,amsthm}
\usepackage{mathrsfs}
\usepackage[shortlabels]{enumitem}
\usepackage{braket}
\usepackage{bbold}
\usepackage{comment,ulem} 
\usepackage{graphicx, xcolor, varwidth}
\usepackage{tikz,tikz-3dplot}
\usepackage{color}
\newtheorem{thm}{Theorem}
\newtheorem{cor}[thm]{Corollary} 
\newtheorem{lem}[thm]{Lemma}
\theoremstyle{remark}
\newtheorem{defn}[thm]{Definition}

\newtheorem{conj}[thm]{Conjecture}

\DeclareMathOperator{\cl }{cl}
\DeclareMathOperator{\setint }{int}
\DeclareMathOperator{\A}{Area}

\renewcommand{\S}{S_{\rm gen}}
\newcommand{\hmax}{H_{\rm max}}

\newcommand{\emax}{e_{\rm max}}

\newcommand{\scri}{\mathscr{I}}

\usepackage{float}

\def\z{z}

\newcommand{\hmg}{H_{\rm max,gen}}

\author{Raphael Bousso and Elisa Tabor}

\affiliation{Center for Theoretical Physics and Department of Physics,\\
University of California, Berkeley, California 94720, U.S.A. 
} 

\emailAdd{bousso@berkeley.edu}
\emailAdd{etabor@berkeley.edu}

\title{Simple Holography in General Spacetimes}

\abstract{The simple or ``outermost'' wedge in AdS is the portion of the entanglement wedge that can be reconstructed with sub-exponential effort from CFT data. Here we furnish a definition in arbitrary spacetimes: given an input wedge $a$ analogous to a CFT boundary region, the simple wedge $z(a)$ is the largest wedge accessible by a ``zigzag,'' a certain sequence of antinormal lightsheets. We show that $z(a)$ is a throat, and that it is contained in every other throat. This implies that $z(a)$ is unique; that it is contained in the generalized entanglement wedge; and that it reduces to the AdS prescription as a special case. 

The zigzag explicitly constructs a preferred Cauchy slice that renders the simple wedge accessible from $a$; thus it adds a novel structure even in AdS. So far, no spacelike construction is known to reproduce these results, even in time-symmetric settings. This may have implications for the modeling of holographic encoding by tensor networks.}

\makeatletter
\gdef\@fpheader{\mbox{}}
\makeatother

\begin{document}
\maketitle

\section{Introduction}

The AdS/CFT correspondence~\cite{Maldacena:1997re} provides a full nonperturbative description of asymptotically Anti-de Sitter spacetimes from the point of view of an external observer, in terms of a conformal field theory defined on conformal infinity, $\scri$. The CFT Hamiltonian is explicitly known in some cases. In order to understand its implications for the gravitating ``bulk,'' a bulk/boundary dictionary must be used before and after any application of the CFT Hamiltonian. This dictionary is straightforward for bulk operators that approach the boundary~\cite{Gubser:1998bc, Witten:1998qj, Susskind:1998dq, Banks:1998dd}, and thus for all operators that are causally accessible from the boundary~\cite{Hamilton:2005ju,Hamilton:2006az}.

The reconstruction of the deeper bulk from CFT data is nontrivial and exhibits a rich structure. It involves sophisticated concepts such as operator quantum error correction~\cite{Almheiri:2014lwa} and single-shot quantum state merging~\cite{Akers:2019wxj}. But it is, by its nature, a kinematic problem that does not involve the CFT Hamiltonian. Deep bulk reconstruction is based on applications of the gravitational path integral~\cite{Lewkowycz:2013nqa}. It involves geometric constructions formulated in the language of classical or semiclassical gravity. As such, it is amenable to generalizations to other spacetimes. 

The notion of the entanglement wedge~\cite{Ryu:2006bv,Hubeny:2007xt,Faulkner:2013ana,Engelhardt:2014gca}---the bulk subregion reconstructible from a boundary subregion $B$---has already been extended to arbitrary spacetimes. The generalized entanglement wedge, or ``hologram,'' is a map~\cite{Bousso:2022hlz,Bousso:2023sya} that takes a \textit{bulk} region as input and outputs a (generically) larger bulk region. Various nontrivial properties of holograms indicate that the output region can in some sense be reconstructed from the input region, and that its generalized entropy represents a true entropy in a fundamental theory~\cite{Bousso:2022hlz,Bousso:2023sya}.

Other critical aspects of bulk reconstruction, though well understood in AdS, have not yet been studied in arbitrary spacetimes. One important concept is that of the simple (or outermost) wedge of a boundary region $B$. The simple wedge $z(B)$ is the homology region between $B$ and the quantum extermal surface closest to $B$ (and hence, from the bulk point of view, outermost). The simple wedge is known to be unique, to contain the causal wedge of $B$, and to be contained inside the entanglement wedge. 

Because the causal wedge can be smaller, the simple wedge contains some information whose reconstruction is not trivial. But it does not contain all reconstructible information, because the entanglement wedge can be larger. The complement of the simple wedge in the entanglement wedge is called the ``Python's lunch.'' 

There is compelling evidence that the reconstruction of operators in the Python's lunch is computationally hard. It requires an exponential (in certain geometric quantities) number of simple logical operations~\cite{Brown:2019rox}. By contrast, semiclassical operators in the simple wedge can be reconstructed from CFT data with polynomial resources~\cite{Engelhardt:2021mue}; hence the name.\footnote{Already in Refs.~\cite{Engelhardt:2017aux,Engelhardt:2018kcs}, which preceded the notion of the Python's lunch, ``simple entropy'' referred to a coarse-grained CFT state whose entanglement wedge is simply reconstructable in the above sense.}

\begin{figure}[h]
    \centering
    \includegraphics[width=16cm]{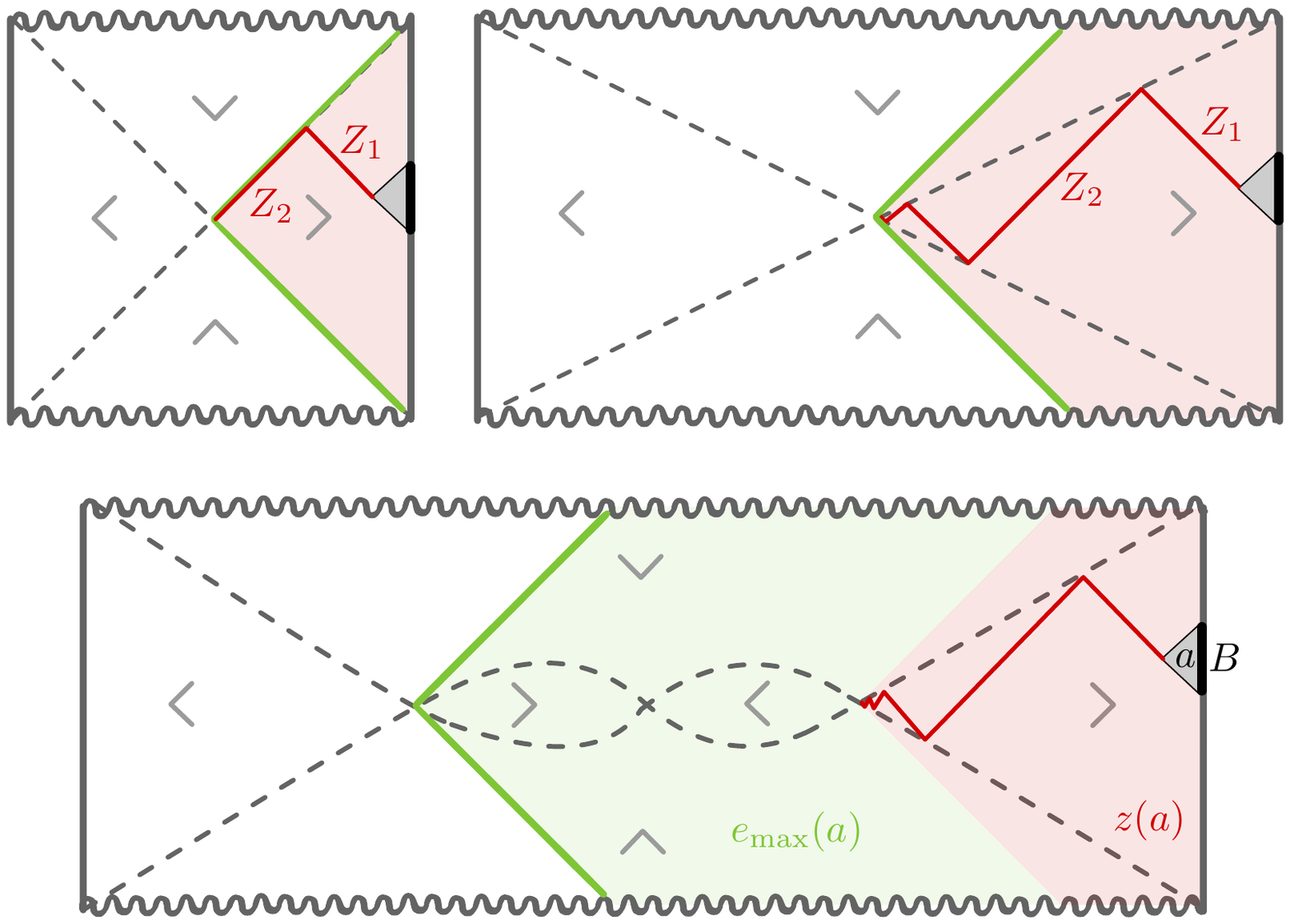}
    \vspace{-.8cm}
    \caption{To start with a familiar setting, each Penrose diagram shows a simple wedge $z(a)$ (shaded red) in AdS. The bulk input region $a$ is shaded grey. The symbols $\wedge$ etc. indicate the null directions in which areas are shrinking~\cite{Bousso:1999cb}. In the top figures, the simple and entanglement wedges agree; the bottom example contains a Python's Lunch. \textit{Top left}: Vacuum Schwarzschild-AdS. The zig $Z_1$ ends at the event horizon; the zag $Z_2$ finds the outermost extremal surface. \textit{Top right}: With generic matter, $Z_n\neq \varnothing$ for all $n$. The zigzag bounces between the future and past apparent horizons (dashed grey), which are PNC and FNC, respectively. \textit{Bottom}: The entanglement wedge is to the right of the solid green lines. Its edge has smaller area than that of $z$, but the zigzag terminates at the outermost throat and thus finds the simple wedge. --- Equivalently (see Sec.~\ref{sec:ads}), each example represents the simple wedge $z(B)$ of a boundary subregion $B$. }
    \label{fig:adszigzag}
\end{figure}

\paragraph{Outline} The purpose of this paper is to provide a definition of the simple wedge in arbitrary spacetimes. We will construct the simple wedge from a sequence of broken lightsheets that we call a zigzag, and we will prove that it obeys the key properties of the simple wedge. In Sec.~\ref{heuristic} we will give a simplified summary of our construction that begins with the classical case before including quantum corrections. We provide several worked examples in figures, for building intuition. In Sec.~\ref{definitions} we provide the definitions and lemmas needed for a rigorous treatment. In Sec.~\ref{full} we present the fully general construction of the simple wedge. We prove its key properties rigorously (by physics standards), and we recover the traditional simple wedge prescription for AdS boundary regions as a special case.

\paragraph{Discussion and Outlook} Our focus in this paper is on the geometric construction of the simple wedge, and on proving the key properties that it is outermost and accessible. In general spacetimes, we do not know the analogue of the CFT, the fundamental theory dual to the semiclassical bulk description. Yet, holograms obey surprising properties such as strong subadditivity of the generalized entropy, which suggest that they capture aspects of states in a true quantum gravity theory. The construction and properties of simple wedges that we establish here suggest that the holographic dictionary will be similar to that of AdS/CFT in even more respects. We hope that this will further constrain and aid our search for the fundamental description of our own universe.

It will be interesting to study whether the interpretation of the simple wedge in terms of computational complexity can be corroborated in general spacetimes. In particular, it would be nice to understand the relation between our zigzag construction and that of Ref.~\cite{Engelhardt:2021mue}. The latter provides an explicit simple recovery protocol. Unlike our construction, it changes the spacetime and is currently defined only perturbatively in general settings, so the prescriptions are not equivalent. 

Turning to the Python's lunch, an important task will be to extend the formula for the exponential reconstruction complexity~\cite{Brown:2019rox,Engelhardt:2023bpv} to general spacetimes. This will require a definition of the relevant ``bulge'' wedges in arbitrary spacetimes. Moreover, at the fully general level of max- and min-reconstruction~\cite{Akers:2019wxj}, no ``Python's Lunch'' prescription for computing the complexity of reconstruction beyond the simple wedge is available even in the original AdS/CFT setting, let alone in our general setting. This remains another key task.

It will also be interesting to explore the implications of our construction for tensor networks~\cite{Swingle:2009bg}. So far, tensor networks have mostly been viewed as discretized versions of the classical geometry of a time-reflection symmetric Cauchy slice. The zigzag slices we construct always leave the time-reflection symmetric slice even if one exists, and this appears to be essential (see Fig.~\ref{fig:horse}). This suggests that perhaps tensor networks should be more broadly viewed as models of broken null hypersurfaces. It will be interesting to investigate whether mean curvature flow~\cite{Nomura:2018kji,Nomura:2018kji} can be used to construct spacelike Cauchy slices that render the simple wedge accessible.\footnote{We would like to thank Guanda Lin and Pratik Rath for an initial exploration of this question.}

\section{Heuristic Summary and Examples}
\label{heuristic}

In this section, we provide a slightly simplified version of the zigzag construction of the simple wedge, and we describe the key properties of the simple wedge. We also provide several instructive examples. We will not provide proofs or overly rigorous definitions; the goal is to offer a first introduction. The reader interested in more rigor is encouraged to study the remaining sections.

\subsection{Classical Simple Wedge in Arbitrary Spacetimes} Recall that a \emph{lightsheet} is a null hypersurface generated by nonexpanding null geodesics orthogonal to a surface~\cite{Bousso:1999cb}. A \emph{wedge} is the full causal development of an open spatial region. A wedge is called \emph{antinormal} on some portion of its edge if both the past and future outgoing orthogonal null congruences are nonexpanding. For example, the spatial exterior of a round sphere in Minkowski space is an antinormal wedge. (Intuitively, the antinormal property indicates that a wedge holographically encodes at least the infinitesimally nearby regions in the antinormal direction.)

We consider a spacetime $M$ that satisfies the classical Einstein equations with matter obeying the Null Energy Condition~\cite{Wald:1984rg}. The starting point of our construction is an arbitrary ``input wedge'' $a\subset M$. The role of $a$ is analogous to a boundary subregion $B$ in AdS, in that we will construct a simple wedge pertaining to $a$.

\begin{figure}[h]
    \centering
    \includegraphics[width=15cm]{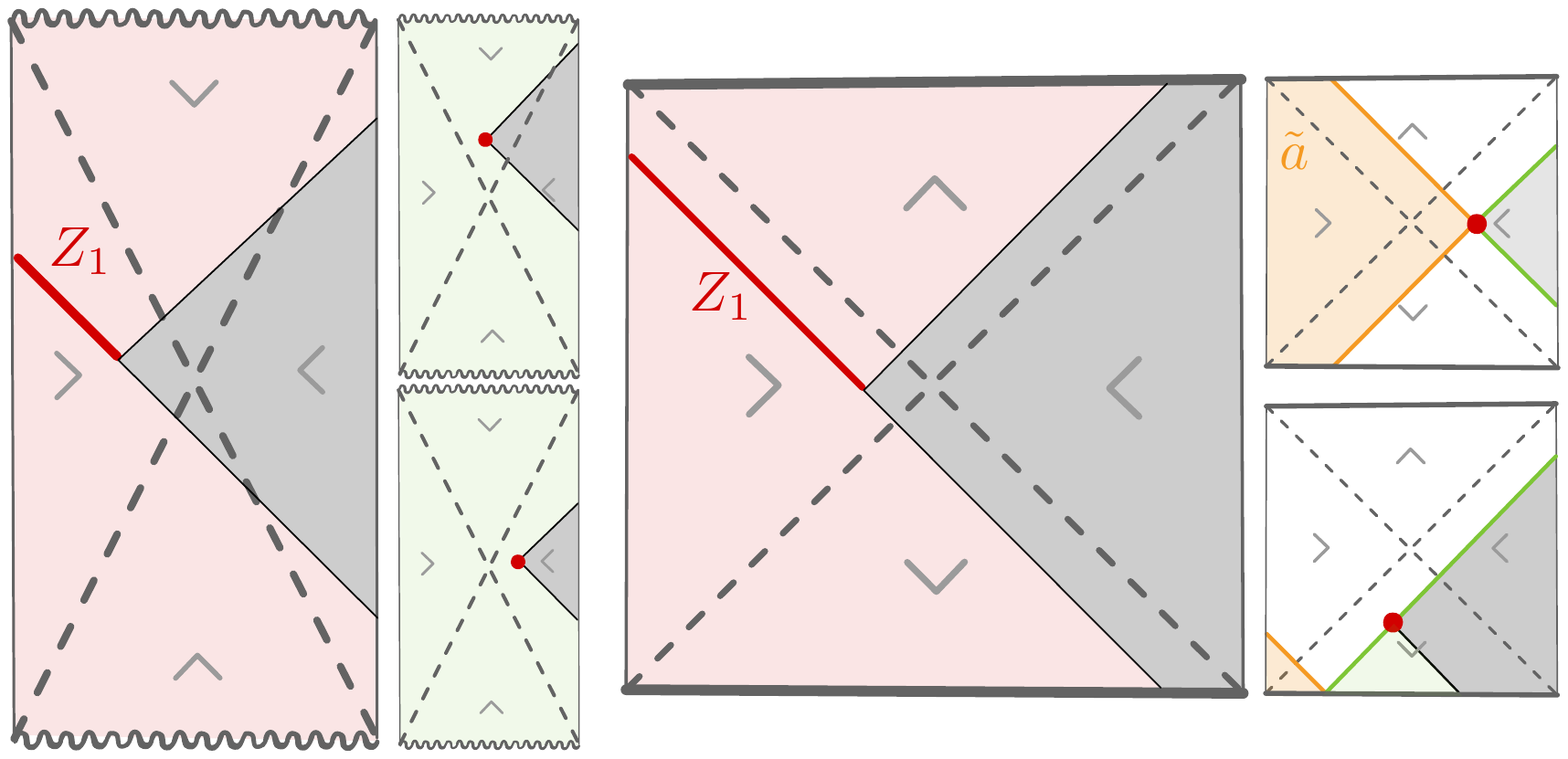}
    \vspace{-3.4cm}
    \caption{Simple wedges in spatially closed universes. The input wedge $a$ is shaded grey, the simple wedge $z(a)$ red, and the max-hologram $\emax(a)$ green. \textit{Left group}: FRLW universe filled with pressureless dust. In all three panels,  the fundamental complement $\tilde a$ is empty and $\emax(a)=M$. In the large figure, $z(a)=M$ is reached by a single zig. In the small figures, $z(a)=a$, and the rest of $M$ is a Python's lunch. \textit{Right group}: vacuum de Sitter space. In the large figure, $\tilde a=\varnothing$ and $\emax(a)=z(a)=M$. In both small figures, $\tilde a\neq \varnothing$ (orange); $z(a)=a$, and $\emax(a)=\tilde a'$. In the top example, $\emax(a)=z(a)$; at the bottom, $\emax(a)$ has a lunch.}
    \label{fig:closed} 
\end{figure}

To construct the first \emph{zig}, $z_1$, we follow the future lightsheet $Z_1$ of $a$ until it ceases to be antinormal; see Fig.~\ref{fig:adszigzag} for an example. If $a$ is nowhere antinormal (as in all of the small panels in Fig.~\ref{fig:closed}), then $z_1=a$. If $a$ is antinormal only on a subset of its edge (as in Fig.~\ref{fig:horse}), then $z_1$ ``grows out'' only from that subset. The zig $z_1$ is defined to be the smallest wedge that contains both $a$ and the lightsheet $Z_1$, i.e., $z_1$ is the causal development of $a\cup Z_1$.

To understand how far $Z_1$ extends, we have indicated the apparent horizons in Fig.~\ref{fig:adszigzag} (dashed lines). Beyond the future apparent horizon, spheres are trapped, so the lightsheet $Z_1$ would fail to be antinormal if it were extended into that region. We must stop at the apparent horizon, where the past directed null expansion vanishes. In Fig.~\ref{fig:closed} (top left), by contrast, the zig extends across the entire complement of $a$. Thus $z_1=M$, and all further steps will be trivial. In Fig.~\ref{fig:horse}, as in Fig.~\ref{fig:adszigzag}, $Z_1$ stops at the first\footnote{This requires a rigorous definition. Roughly, $Z_1$ must not contain any future lightsheet of $a$ that is everywhere past-expanding. See Def.~\ref{def:zig} for details.} surface where the past-directed null congruence has vanishing expansion.

\begin{figure}[h]
    \centering
    \includegraphics[width=16cm]{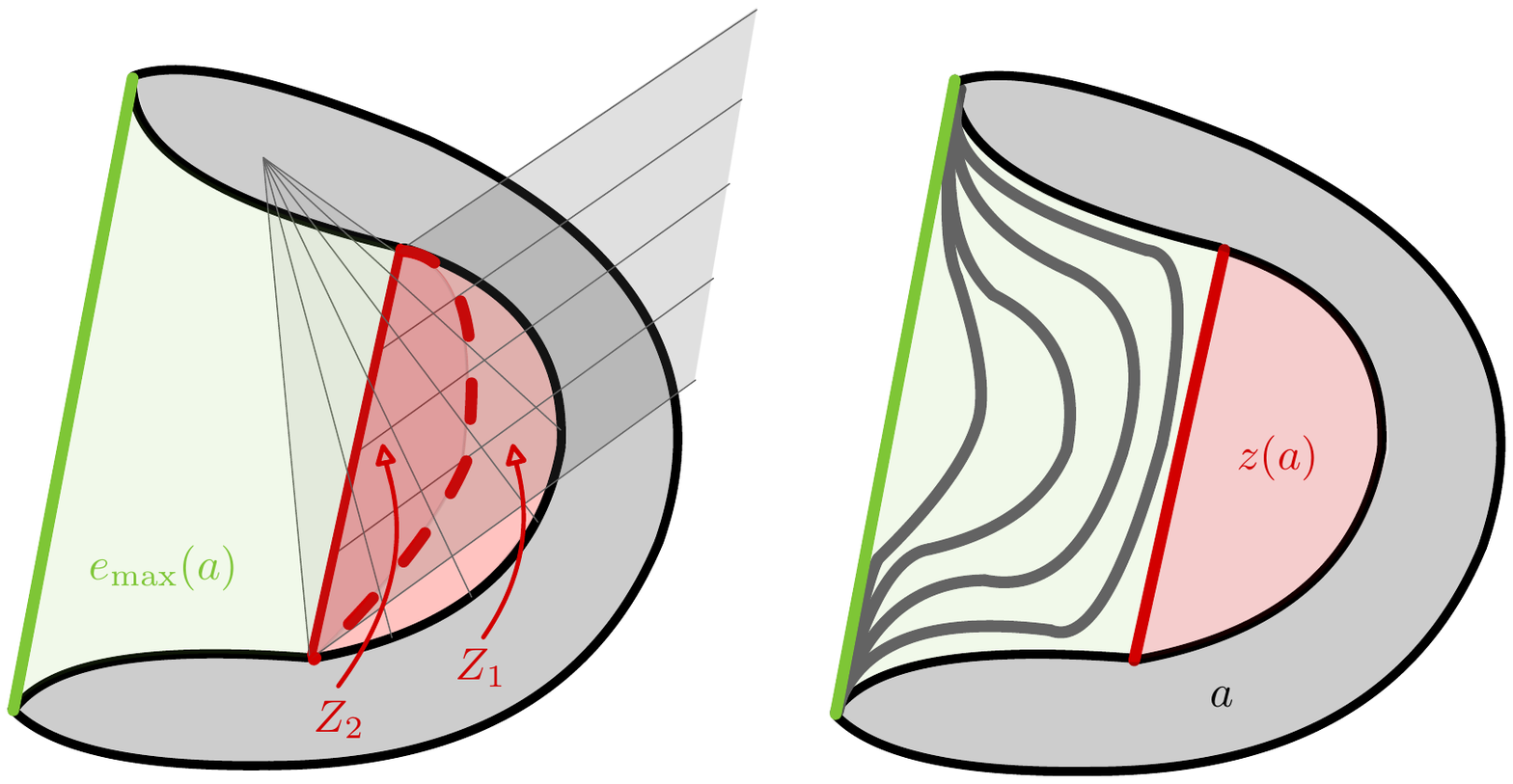}
    \vspace{-3.8cm}
    \caption{Simple wedge $z(a)$ (red edge) and max-entanglement wedge $\emax(a)$ (green edge) of a horseshoe-shaped input wedge $a$ in 2+1 dimensional Minkowski space. \textit{Left}: zigzag construction of $z(a)$. $Z_1$ is a portion of a light cone; it ends on the red dashed line, which is PNC. $Z_2$ is a portion of a null plane that connects this to the extremal red line. No lightsheets begin at the convex (expanding) portions of the horseshoe. \textit{Right}: Although the horseshoe lies on a time-reflection symmetric Cauchy slice $\Sigma$, we have not found a viable construction of $z(a)$ purely on $\Sigma$. In particular, the presence of a lunch is obscured by the existence of a spacelike foliation on $\Sigma$ (grey lines) with monotonically decreasing area, which interpolates from the red to the green edge.}
    \label{fig:horse}
\end{figure}

Having constructed the zig, we now iterate, switching future and past at every step. That is, we follow $Z_2$, the past-directed lightsheet of $z_1$, until it fails to be antinormal. See Figures~\ref{fig:adszigzag} and \ref{fig:horse} for nontrivial examples. The resulting region (including $z_1$) is called the \emph{zag} of $z_1$ and is labeled $z_2$. Continuing onward, alternating between zigs and zags, we construct ever growing regions, $z_3, z_4, \ldots$. We call $z_n$ the $n$-th \emph{zigzag} of $a$. Our construction equips $z_n$ (really, just the portion exterior to $a$) with a preferred Cauchy slice, $Z_1\cup Z_2 \cup \ldots \cup Z_n$.

The \emph{simple wedge} $z$ is the $n\to\infty$ limit of the zigzag. Our Corollary~\ref{cor:zprop} reduces, in the classical setting, to the proof that $z$ is a stationary surface for the area functional\footnote{In a Lorentzian geometry, no surface has locally minimal or maximal area, so none is extremal. Nevertheless we will use ``extremal'' instead of ``stationary'' below, since this terminology is widely used.} (i.e., has vanishing past and future expansions) where its edge differs from that of $a$. Theorem~\ref{zthm:zigoutermost} reduces to the proof that $z$ is contained inside all other wedges with this property. These two results establish that $z$ is properly called the simple or outermost wedge.

\subsection{Quantum Extensions}

The Quantum Extremal Surface (QES) prescription~\cite{Engelhardt:2014gca} replaces the area by the generalized entropy $\S$ (the Bekenstein-Hawking entropy of the edge plus the von Neumann entropy of the matter fields in the wedge). This evades the need for assuming the Null Energy Condition (which is false in nature), and it is vital for the correct treatment of semiclassical phenomena such as black hole evaporation. The improvement is analogous to replacing Hawking's area theorem~\cite{Hawking:1971vc} by the Generalized Second Law~\cite{Bekenstein:1972tm}, or the classical focusing of lightrays by the Quantum Focusing Conjecture~\cite{Bousso:2015mna}.  

To implement this substitution, our definition of the simple wedge follows exactly the same steps as outlined above, with the classical expansion replaced by the quantum expansion (a functional derivative of $\S$). In proofs, the Null Energy Condition replaced by the Quantum Focusing Conjecture. An interesting simple wedge that illustrates this regime is constructed in Fig.~\ref{fig:island}, for an evaporating black hole.

The QES prescription was further refined~\cite{Akers:2019wxj} by replacing the von Neumann entropy by the smooth conditional max entropy~\cite{RenWol04a}, $\hmax^\epsilon(b|a)$. This quantity involves two nested wedges $b\supset a$. Unlike the conditional von Neumann entropy, it cannot be expressed as a difference of two (unconditional) entropies associated to $b$ and $a$ separately. Therefore, the max-entanglement wedge cannot be defined as the stationary surface of a single functional such as $\S$. Its construction takes a substantially different form~\cite{Akers:2023fqr,Bousso:2023sya}, from which the old QES prescription can nevertheless be recovered when von Neumann entropies provide a good approximation. (A larger ``min''-entanglement wedge can also be defined, though we will not consider it here.) At the most general setting of holograms (entanglement wedges in arbitrary spacetimes), one finds that the max- and min-hologram can already differ at the classical level~\cite{Bousso:2023sya}. 

\begin{figure}[h]
    \centering
    \includegraphics[width=12cm]{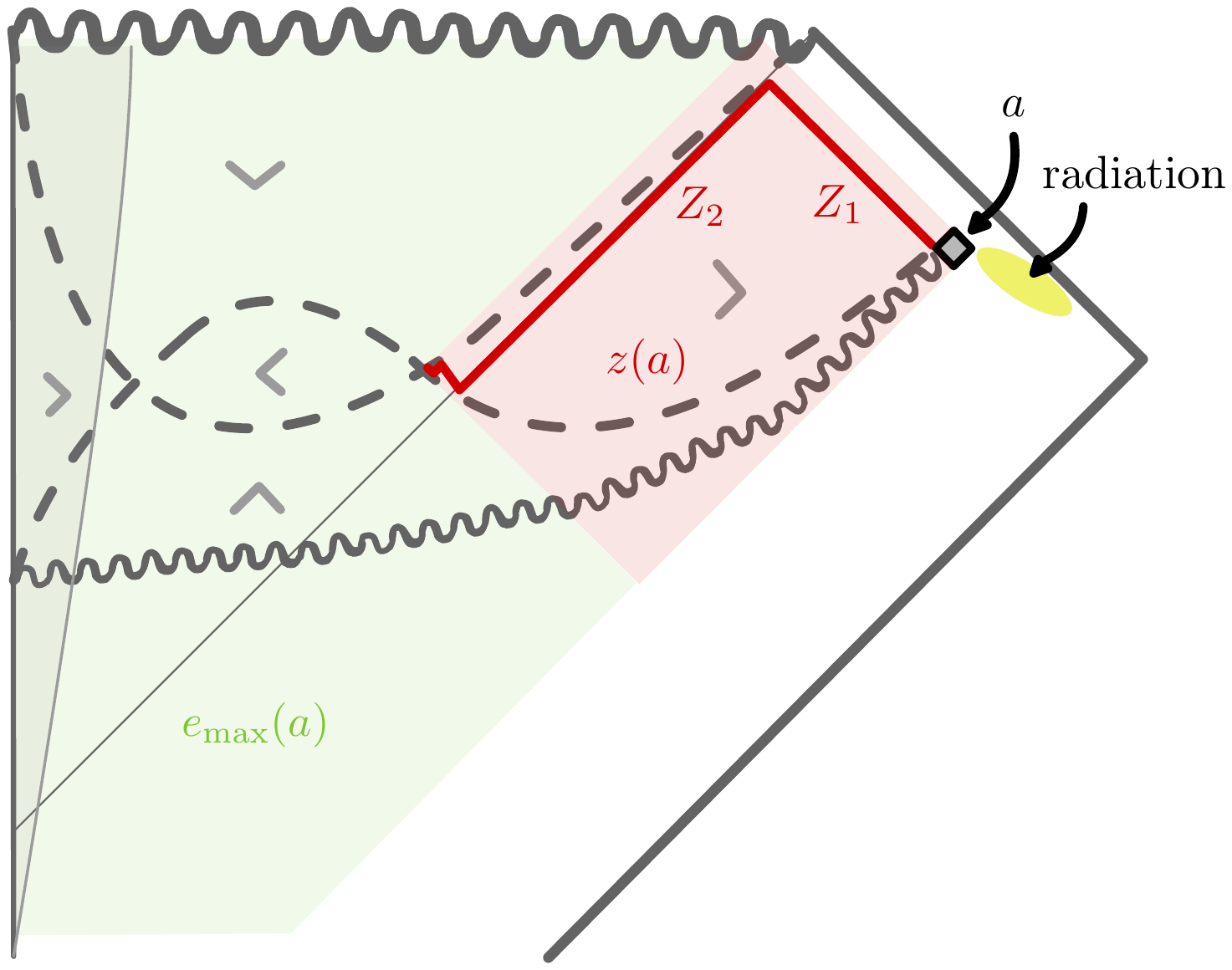}
    \vspace{-.5cm}
    \caption{Penrose diagram of an evaporating black hole~\cite{Bousso:2022tdb,Bousso:2025xyc}. (The symbols $\wedge$ etc. now indicate the null directions in which the generalized entropy $\S$ is shrinking.) The input wedge $a$ is a thin shell located at 10 times the Schwarzschild radius; nearly all of the Hawking radiation emitted so far has propagated to larger radius. The simple wedge extends to a quantum extremal surface located near the black hole horizon~\cite{Penington:2019npb,Almheiri:2019hni}. Before the Page time, the entanglement wedge includes the black hole interior as a Python's lunch; this is the case shown here. After the Page time, $\emax(a)=z(a)$.}
    \label{fig:island}
\end{figure}

Full generality also requires us to allow regions bounded by corners or folds. Corners arise generically at caustics of null congruences, and when considering unions or intersections of regions.\footnote{To see the importance of this level of generality, consider the smooth ($C^\infty$) horseshoe-shaped input region in Fig.~\ref{fig:horse}. Its simple wedge has corners: the null normal vectors are discontinuous. Its entanglement wedge has no corners but discontinuous null expansions.} This leads to a minimal, more robust structure, in which quantitative null expansions are abandoned.  Nonexpansion and noncontraction become qualitative properties, and the Quantum Focusing Conjecture~\cite{Bousso:2015mna} is replaced by a pared-down version, Discrete Max Focusing~\cite{Bousso:2024iry}. Discrete nonexpansion still allows for key definitions, such as the concept of antinormal; and Discrete Max Focusing suffices in proofs.

Working at this level of generality, the appropriate generalization of the notion of quantum extremal surface is the statement that the max-hologram $\emax(a)$ is a ``throat accessible from $a$.''  A throat is an antinormal wedge that cannot be enlarged while staying antinormal. Accessibility from $a$ relaxes the antinormal requirement to apply only to edge portions of the throat that differ from the edge of $a$. But accessibility adds a new condition: that the portion of the throat wedge outside of $a$ admits a partial Cauchy slice such that the generalized max entropy of the throat conditioned on any intermediate wedge whose edge lies on the slice is negative. (Classically, this reverts to the statement that all intermediate surfaces on the Cauchy slice have larger area than the throat.)

Our definition of the simple wedge $z$ is fully general in that it incorporates all of the above refinements and generalizations. This requires the apparatus of definitions and lemmas developed in the following section.


\section{Preliminary Definitions and Lemmas}
\label{definitions}

Here we reproduce a number of standard definitions, fix notation, and derive some Lemmas.

\subsection{Wedges and Causal Structure}
\label{sec:wedges}

Let $M$ be a globally hyperbolic Lorentzian spacetime with metric $g$. The chronological and causal future and past, $I^\pm$ and $J^\pm$, and the unphysical spacetime $M\cup\partial M$ with conformal boundary $\partial M$ are defined as in Wald~\cite{Wald:1984rg}.

A proper subset will be denoted by $\subsetneq$; $\subset$ permits equality. For $s\subset M$, $\setint s$, $\cl s$, and $\partial s$ denote the interior, the closure, and the boundary of $s$. All operations are performed in $M$ unless explicitly stated otherwise; in those cases we will denote the relevant set by a subscript. 

\begin{defn}\label{zdef:sc}
The \emph{spacelike complement} of a set $s\subset M$ is
\begin{equation}\label{zeq:sc}
    s'=M\setminus \cl[I(s)]~.
\end{equation}
(Thus, $s'$ is necessarily open.)
\end{defn}

\begin{defn}\label{zdef:covwedge}
A {\em wedge} is a set $a\subset M$ that satisfies $a=a''$.  (Thus $a$ is an open set; $a'$ is a wedge; and the intersection of two wedges $a,b$ is a wedge~\cite{Bousso:2022hlz, Bousso:2023sya}.)
\end{defn}

\begin{defn}\label{zdef:wedgeunion}
The {\em wedge union} of two wedges $a,b$ is the wedge
\begin{equation}
    a\Cup b\equiv (a'\cap b')'~.
\end{equation}
\end{defn}

\begin{defn}\label{zdef:edgehor}
The \emph{edge} $\eth a$ and \emph{Cauchy horizons} $H^\pm(a)$ of a wedge $a$ are
\begin{align}
    \eth a & \equiv \partial a \setminus I(a)~,\\
    H^+(a) & \equiv \partial a\cap I^+(a)~,\\
    H^-(a) & \equiv \partial a\cap I^-(a)~,\\
    H(a) & \equiv \partial a\cap I(a) = H^+(a)\cup H^-(a)~.
\end{align}
\end{defn}

\begin{defn}[Null Infinity]
    \emph{Future infinity}, $\scri^+$, is the subset of $\partial M$ consisting of the future endpoints in $\bar M$ of null geodesics of future-infinite affine length in $M$. \emph{Past infinity}, $\scri^-$, is defined similarly. \emph{Null infinity} is their union:
    \begin{equation}
        \scri \equiv \scri^+ \cup \scri^-~.
    \end{equation}
\end{defn}

We adopt the following two definitions from Ref.~\cite{BoussoKaya}:
\begin{defn}
The \emph{conformal shadow} $\tilde \scri$ of a wedge $a$ is the subset of $\scri$ that is causally inaccessible from $a$: 
\begin{equation}
    \tilde \scri(a) \equiv  \scri \setminus [\cl I(a)]_{\bar M}~ .
\end{equation}
The subscript tells us that $\cl I(a)$ should be computed in $\bar M\equiv M\cup \scri$. 
\end{defn}

\begin{defn}\label{def:tildea}
The \emph{fundamental complement} $\tilde a$ of a wedge $a$ is the double complement of the conformal shadow of $a$. More precisely, 
\begin{equation}
    \tilde a \equiv  (\tilde \scri(a)'_{\bar M}\cap M)'~.
\end{equation}
\end{defn}

\subsection{Discrete Nonexpansion, Accessibility, Holograms, and Lightsheets}

The following treatment is somewhat simplified in order to avoid distractions. We will not display the smoothing parameter $\epsilon$ in the smooth conditional max- and min-entropies. In generalized entropies, we treat the area terms as separable from the matter entropy terms, and we omit terms of order $G$ and higher. For more careful definitions of smoothing and of the generalized max-entropy, see Ref.~\cite{Akers:2023fqr}. For a discussion of the $G\hbar$ expansion, we refer the reader to Refs.~\cite{Shahbazi-Moghaddam:2022hbw, Bousso:2024iry}.

The central quantitative object we consider is the generalized smooth conditional max-entropy of two wedges,
\begin{equation}\label{eq:hmgdef}
    \hmg(b|a) \approx \left[ \frac{\A(b)-\A(a)}{4G} + \hmax(b|a) + O(G)\right]~.
\end{equation}
Here $\A(a)$ is the area of $\eth a$, and $\hmax(b|a)$ is the smooth max-entropy of the quantum fields in $b$ conditioned on $a$~\cite{RenWol04a}. For compressible quantum states (which are often considered), this becomes a difference of von Neumann entropies:
\begin{equation}
    \hmax(b|a) \approx S(b|a) \equiv S(b)-S(a)~.
\end{equation}
In this approximation,
\begin{equation}
    \hmg(b|a) \approx \S(b|a) \equiv \S(b)-\S(a)~,
\end{equation}
where 
\begin{equation}
    \S(a)\equiv \frac{\A(a)}{4G} + S(a) + O(G)
\end{equation}
is the generalized entropy~\cite{Bekenstein:1972tm} of the wedge $a$.
In many interesting settings, we will be able to neglect the contribution of bulk quantum fields entirely and keep only the area terms.

\begin{thm}\label{zconj:SSA}
    $\hmg$ satisfies \emph{strong subadditivity}: let $a,b,c$ be bulk subregions such that $a\supseteq b,c$ and $a\cap b'\subseteq c$. Then
    \begin{equation}\label{eq:SSA}
        \hmg(a|c)\leq\hmg(b|c\cap b)\,.
    \end{equation}
    Here we will only need a weaker statement, \emph{Discrete Subadditivity}~\cite{Bousso:2024iry}:
    \begin{equation}\label{eq:SSA}
        \hmg(b|c\cap b)\leq 0 \implies \hmg(a|c)\leq 0~.
    \end{equation}
\end{thm}

\begin{thm}\label{zthm:chainrule}
    $\hmg$ also satisfies a \emph{chain rule}: Let $a,b,c$ be bulk subregions such that $a\supset b\supset c$. Then
    \begin{equation}\label{eq:chainrule}
        \hmg(a|c) \leq\hmg(a|b)+\hmg(b|c)\,.
    \end{equation}
\end{thm}
At the level of approximation of Eq.~\eqref{eq:hmgdef}, both theorems follow immediately from the fact that the areas and the conditional matter entropies satisfy the corresponding properties~\cite{Konig_2009, Vitanov_2013}. Eq.~\eqref{eq:SSA} was proven without assuming the separability of the area term in Ref.~\cite{Akers:2023fqr} (Proposition 3.13). At higher orders in $G$, the theorems must be re-stated as conjectures~\cite{Bousso:2024iry}.

\begin{defn}[Nonexpanding wedges]
    A wedge $a$ is said to be \emph{future-nonexpanding (FNE)}~\cite{Bousso:2024iry} at $p\in\eth a$ if there exists an open set $O$ containing $p$ such that 
    \begin{align}\label{eq:futnonexp}
        \hmg(b|a) \leq 0 \text{~for~all~wedges~}b\supset a\text{~such~that~} 
        \eth b \subset \eth a \cup [H^+(a')\cap O]~.
    \end{align}
    To define \emph{past-nonexpanding (PNE)}, $H^-$ replaces $H^+$ in Eq.~\eqref{eq:futnonexp}. The wedge $a$ is said to be \emph{antinormal} at $p\in\eth a$ if $a$ is FNE and PNE at $p$.
\end{defn}

See Figure~\ref{fig:ads} and Appendix~\ref{app:examples} for an illustration of nonexpansion and of other concepts introduced below. When a numerical outward future quantum expansion is well-defined, FNE implies its nonpositivity; moreover, negative quantum expansion implies FNE~\cite{Bousso:2024iry}. However, we shall not require numerical values here, and the above definition is superior because it can be applied to non-smooth edges (which arise generically).

\begin{figure}[t]
    \vspace{-1cm}
    \hspace{-1cm}
    \includegraphics[width=18cm]{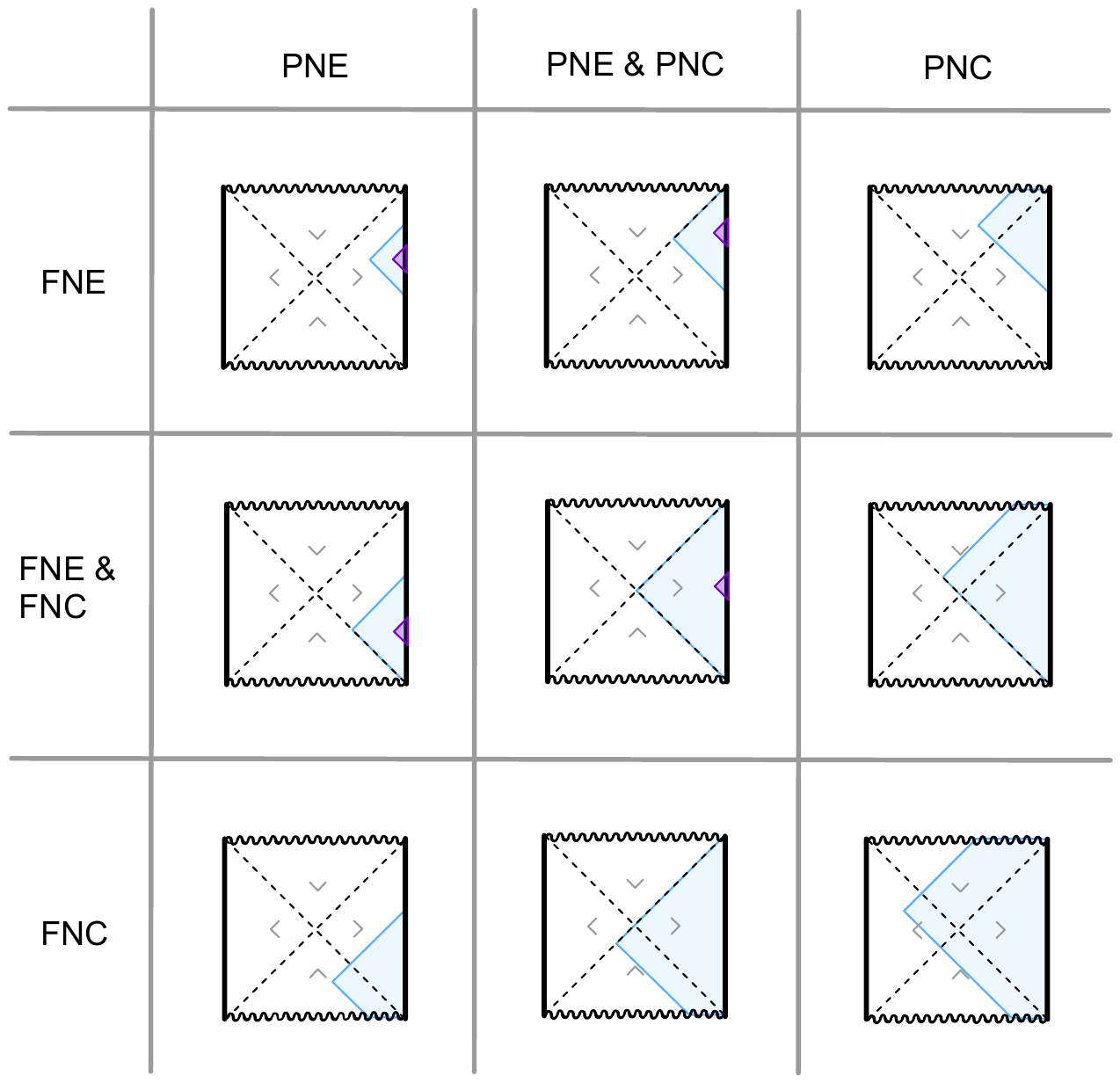}
    \vspace{-1cm}
    \caption{Illustration of the notions of noncontraction, nonexpansion, and accessibility using examples in Schwarzschild-AdS. In the cases where the blue wedge is antinormal (i.e., FNE and PNE), it is accessible from the purple wedge (clockwise from top left: accessible; past-marginally accessible; throat accessible; future-marginally accessible).}
    \label{fig:ads}
\end{figure}

\begin{defn}[Accessibility]
  Given a wedge $a$, the wedge $k\supset a$ is said to be \emph{accessible from $a$} if it satisfies the following conditions:
\begin{enumerate}[I.]
    \item $a\subset f\subset \tilde a '$, where $\tilde a$ is the fundamental complement of $a$ (see Def.~\ref{def:tildea});
    \item $k$ is antinormal at points $p\in \eth f\setminus\eth a$;
    \item $k$ admits a Cauchy slice $\Sigma$ such that $\Sigma\supset \eth a$ and such that for any wedge $h\subsetneq k$ with $a\subset h$, $\eth h\subset \Sigma$, and $\eth h\setminus \eth k$ compact in $M$,\footnote{Some references require a strict inequality~\cite{Bousso:2023sya,Akers:2023fqr}, but our choice will be convenient in proofs. The distinction is not meaningful given the smoothing inherent in the definition of the conditional max entropy $\hmg^\epsilon$.}
    \begin{equation}
        \hmg(k|h)\leq 0~.
    \end{equation}
\end{enumerate}
Operationally, condition III is the statement that the area of any intermediate surface $\eth h$ suffices as an entanglement resource for performing quantum state merging from $k$ to $h$~\cite{Akers:2019wxj}.  
\end{defn} 

\begin{defn}[Max-Hologram]\label{def:emax}
    Given a wedge $a$, its \emph{max-hologram} (or \emph{generalized max-entanglement wedge}), $\emax(a)$, is the wedge union of all wedges that are accessible from $a$~\cite{Bousso:2023sya}.
\end{defn}

\begin{defn}[Lightsheets]
    Let \emph{$\eth a^+$} be the set of points where the wedge $a$ is FNE. A null hypersurface \emph{$L^+(a)$} $\subset H^+(a')$ whose past boundary lies on $\eth a^+$ is called a \emph{future lightsheet} of $a$~\cite{Bousso:1999xy,Bousso:1999cb}. The outward deformation of $a$ along the future lightsheet $L^+(a)$, $a\Cup L^+(a)$, is called a \emph{future lightsheet wedge} of $a$. Past lightsheets and past lightsheet wedges are defined analogously.
\end{defn} 

The following conjecture is a minimal, discrete version~\cite{Shahbazi-Moghaddam:2022hbw,Bousso:2024iry} of the Quantum Focussing Conjecture~\cite{Bousso:2015mna,Bousso:2015wca} (which in turn becomes the focusing theorem of General Relativity in the classical limit):
\begin{conj}[Discrete Max-Focusing]\label{conj:qfc}
    Let $b\subset c$ both be future lightsheet wedges, or both be past lightsheet wedges, of $a$. Then 
    \begin{equation}\label{eq:qfc}
        \hmg(c|b)\leq 0~.
    \end{equation}
An immediate consequence is the ``persistence of nonexpansion'': for any future lightsheet wedge $b$ of $a$
    \begin{equation} \label{eq:persistence}
        \eth b \cap J^+(\eth a^+) \subset \eth b^+~.
    \end{equation}
\end{conj}

\subsection{Discrete Noncontraction, Marginal Wedges, and Throats}

\begin{defn}[Noncontracting Wedges]
    A wedge $a$ is said to be \emph{future-noncontracting} (FNC) if there exists an open set $O\supset\eth a$ such that no proper past-directed outward null deformation of $a$ with compact support within $O$ is FNE on all new edge points. That is, no wedge $b\supsetneq a$ with $\eth b\subset \eth a \cup [H^-(a')\cap O]$ is FNE on all points in $\eth b\setminus\eth a$. (Heuristically, in sufficiently smooth settings, FNC implies nonnegativity of the outward future quantum expansion of $a$; moreover, positivity of the quantum expansion implies FNC.) \emph{Past-noncontracting} (PNC) is defined analogously.   
\end{defn}

\begin{lem} \label{zlem:intersectFNC}
    The intersection of two FNC (PNC) wedges is FNC (PNC).
\end{lem}
\begin{proof}
    Let $a$ and $b$ be FNC wedges, and suppose for contradiction that $c\equiv a\cap b$ is not FNC. Then for any open set $O_c\supset\eth c$, there exists a past deformation $d\supsetneq c$, with $\eth d \subset\eth c\cup [H^-(c')\cap O_c]$, such that $d$ is everywhere FNE on $\eth d\setminus\eth c$. 
    
    Let $f=d\Cup a$. We may assume that $f\neq a$; otherwise exchange the names of $a$ and $b$. Since $d$ is everywhere FNE on $\eth d\setminus\eth c$, Discrete Max-Focusing and Discrete Subadditivity imply that $f$ is everywhere FNE on $\eth f\setminus\eth a$. 
    Moreover, we can construct such a deformation in any open neighborhood $O_a$ of $a$ by choosing $O_c\subset a\Cup O_a$. This shows that $a$ is not FNC, in conflict with the assumption of the Lemma. Hence $c$ must be FNC.

    The time-reversed argument shows that $a\cap b$ is PNC for PNC wedges $a$ and $b$.
\end{proof}


\begin{defn}[Marginal accessibility from $a$]
Let $a\subset k$ be wedges. $k$ is called \emph{future-marginally accessible} from $a$ if $k$ is accessible from $a$ and $k$ is FNC. \emph{Past-marginally accessible} is defined analogously. The wedge $k$ is called a \emph{throat accessible from $a$} if $k$ is future- and past-marginally accessible from $a$.     
\end{defn}

\section{Zigzag and the Simple Wedge}
\label{full}

\subsection{Simple Wedge in General Spacetimes}
\label{sec:general}

\begin{defn}[Zig and zag]\label{def:zig}
Given a wedge $a$, let $Q(a)$ be the set of future lightsheet wedges $l^+(a)$ that satisfy 
\begin{enumerate}[A.]
    \item $l^+(a)$ is PNE---and hence antinormal, by Eq.~\eqref{eq:persistence}---on $\eth l^+(a)\setminus \eth a$;
    \item $l^+(a)$ contains no PNC future lightsheet wedge of $a$ as a proper subset;
    \item $l^+(a)\subset\tilde a'$\footnote{Property C ``almost'' follows from the rest of the definition, in the sense that it could be eliminated with a weak ``generic'' assumption analogous to an assumption that prevents null congruences from having vanishing expansion over finite affine length in classical General Relativity.}.
\end{enumerate}
We define the \emph{zig} of $a$, $z^+(a)$ as their wedge union:
\begin{equation}
    z^+(a) = \Cup_{l^+(a)\in Q(a)} ~l^+(a)~.
\end{equation}
The \emph{zag} $z^-(a)$ is similarly defined in terms of past lightsheet wedges of $a$ that contain no FNC past lightsheet wedge of $a$ as a proper subwedge.
\end{defn}

Next, we establish two important properties of the zig.

\begin{lem}
    $z^+(a)\in Q(a)$.
\end{lem}

\begin{proof}  
Property A: By Corollary 48 of Ref.~\cite{Bousso:2024iry}, $z^+(a)$, too, is antinormal on $\eth z^+(a)\setminus \eth a$.

Property B: Suppose $z^+(a)$ contained a proper PNC subwedge $w$ that was a future lightsheet wedge of $a$. Therefore $Q(a)$ contains an element $\hat l(a)$ that is not contained in $w$ (or else $w$ would not be a proper subwedge of $z^+(a)$. The union of two lightsheet wedges is obviously itself one: $v\equiv w\Cup \hat l(a)\in Q(a)$; yet $v$ contains the PNC wedge $w$ as a proper subwedge. This contradicts the definition of $Q(a)$.

Property C holds trivially.
\end{proof}

\begin{lem}\label{lem:pnc}
$z^+(a)$ is PNC (in the spacetime $\tilde a'$).
\end{lem}

\begin{proof}
    If $z^+(a)$ were not PNC, then there would exist a future outward deformation of $z^+(a)$ that is PNE and which thus would enlarge the wedge union.
\end{proof}

The zag obeys the time-reversed properties; thus the zag is FNC. Clearly $z^+(z^+((a))=z^+(a)$. But when zigs are alternated with zags, the wedge can keep growing. 

\begin{defn}\label{def:zigzag}
We define the \emph{$n$-th zigzag} of $a$ inductively, by setting $z_0=a$ and 
\begin{equation}
    z_n(a) \equiv z^+(z_{n-1}(a))~~(n~\text{odd})~~;~~~
    z_n(a) \equiv z^-(z_{n-1}(a))~~(n~\text{even})~.
\end{equation}
The zigzag construction also defines a preferred, piecewise null Cauchy slice for $z_n(a)\cap a'$, defined inductively by setting $Z_0=\varnothing$ and
\begin{equation}
    Z_n(a)= H[(z_n(a)]\setminus I[z_{n-1}(a)]~.
\end{equation}
See Fig.~\ref{fig:adszigzag} for an illustration.
\end{defn}

\begin{thm}[Accessibility of the zigzag] 
\label{zthm:accessible}
   For all $n$, the $n$-th zigzag of $a$, $z_n(a)$, is accessible from $a$ via the Cauchy slice $Z_n(a)$.\footnote{In the classical limit, the stronger statement of monotonicity holds: $\A(h)- \A(g)\leq 0$ for any two intermediate wedges $z\supset h\supset g \supset a$ with edges on $Z_n(a)$. However, this stronger statement does not hold at any quantum level. The fact that monotonicity already fails at the level of $\S$ is a premonition of the role of max and min entropies in a structure that ignores them.}
\end{thm}
\begin{proof}
    We must show that $z_n(a)$ satisfies properties I-III of an accessible wedge.

    \textit{Property I} $a\subset z_n(a)\subset\tilde a'$ by construction. 
    
    \textit{Property II:} $z_n(a)$ is antinormal at points $p\in\eth z_n(a)\setminus\eth a$ since it is the union of sets with the same property (Corollary 48 in Ref.~\cite{Bousso:2024iry}).

    \textit{Property III:} The Cauchy slice $Z_n(a)$ satisfies $Z_n(a)\supset\eth a$. We must show that
    \begin{equation}
        \hmg[z_n(a)|h] \leq 0
    \end{equation}
    for any wedge $h\neq z_n(a)$ such that $h\supset a$, $\eth h\subset Z_n(a)$, and $\eth h\setminus\eth z(a)$ is compact in $M$. Since only one input wedge $a$ is involved, we suppress the arguments ``$(a)$'' in the derivation below:
    \begin{align}
        &\hmg[z_{n}|h\Cup z_{n-1}] \leq 0 \\
        &\hmg[h\Cup z_{n-1}|h\Cup z_{n-2}] \leq 0 \\
        &\qquad \vdots \nonumber\\
        &\hmg[h\Cup z_{1}|a] \leq0 \,.
    \end{align}
    To obtain the first inequality, note that $h\Cup z_{n-1}$ is a null deformation of the antinormal wedge $z_{n-1}$ and thus nonexpanding towards the edge of $z_n$, by the persistence of nonexpansion. The first inequality then follows from Discrete Max-Focusing (Conj.~\ref{conj:qfc}). 
    
    A completely analogous argument yields $\hmg[z_{n-1}|(h\cap z_{n-1})\Cup z_{n-2}]\leq 0$. By taking the union of both sets with $h$ and appealing to Discrete Subadditivity (Conj.~\ref{zconj:SSA}) we obtain the second inequality. This process continues until we reach $a$.
    
    After adding all inequalities and applying the chain rule (Theorem~\ref{zthm:chainrule}), we obtain
    \begin{equation}
        \hmg[z_n|h] \leq 0~.
    \end{equation}
\end{proof}
\begin{cor}
    For all $n$, $z_n(a) \subset \emax(a)$ by Def.~\ref{def:emax}.
\end{cor}
\begin{cor}
    For all odd (even) $n$, $z_n(a)$ is past (future) marginally accessible from $a$, by Lemma~\ref{lem:pnc}.
\end{cor}

\begin{defn}[Simple wedge] \label{zdef:simple}
    The \emph{simple wedge}, $\z(a)$, is defined as an infinite zigzag:
    \begin{equation}
        z(a) \equiv \lim_{n\to\infty} z_n(a)~.
    \end{equation}
    It is equipped with a preferred, piecewise null Cauchy slice
    \begin{equation}
        Z(a)\equiv \lim_{n\to\infty} Z_n(a)~.
    \end{equation}
\end{defn}
Note that $z(a)$ shares all properties of the finite-$n$ zigzags. It is accessible from $a$; and it is both PNC and FNC: otherwise, it could be enlarged by either another zig or another zag. We summarize these properties in a third corollary of Theorem~\ref{zthm:accessible}:
\begin{cor}\label{cor:zprop}
The simple wedge satisfies the following properties:
\begin{itemize}
    \item $z(a)$ is a throat accessible from $a$ via the Cauchy slice $Z(a)$;
    \item in particular, $z(a)$ is antinormal on $\eth z(a)\setminus\eth a$, and $z(a)\subset \emax(a)$.
\end{itemize}
\end{cor}
Moreover, $z(a)$ shares a defining characteristic of the simple (or outermost) wedge in AdS/CFT:
\begin{thm}\label{zthm:zigoutermost}
    $z(a)$ is contained in any other throat accessible from $a$. 
\end{thm}
\begin{proof}
    Suppose for contradiction that $k\supset a$ is a throat but $z(a)\not\subset k$. Then $H(k)\cap z(a)$ is not empty. Notice $\eth k\not\subset Z(a)$, otherwise we have an immediate contradiction with the definition of the zigzag. 
    
    Let $N$ be the smallest $i$ such that $Z_i(a)$ intersects $H(k)$ and such that $z_i(a)\not\subset k$. For odd (even) $N$, $z_N(a)$ is PNC (FNC). By Lemma~\ref{zlem:intersectFNC}, for $N$ odd (even)
    \begin{equation}
        j\equiv k\cap z_N(a)
    \end{equation}
    defines a PNC (FNC) proper subwedge of $z_N(a)$. This contradicts the definition of the zigzag.
\end{proof}
Our definition of the zigzag involved an arbitrary choice: starting from $a$, we first constructed a zig, then a zag, and so on. One could also have defined a ``zagzig'' by starting with a zag instead. At finite $n$ these sets will generically differ, but in the limit as $n\to\infty$ they define the same simple wedge:
\begin{cor}
    $z(a)=z^T(a)$, where $z^T(a)$ is the $n\to\infty$ limit of the ``zagzig'' $z_n^T(a)$ obtained by starting with a zag instead of a zig in Def.~\ref{def:zigzag}.
\end{cor}
\begin{proof}
    It is easy to check that Corollary~\ref{cor:zprop} and Theorem~\ref{zthm:zigoutermost} apply to $z^T(a)$ as well. Thus both $z(a)$ and $z^T(a)$ are throats accessible from $a$, and each must be contained in the other.
\end{proof}

\subsection{Simple Wedge of a Boundary Subregion in AdS}
\label{sec:ads}

Given the historical importance of the AdS/CFT correspondence in the development of the notions of entanglement wedge and simple/outermost wedge, we want our definitions and results to accommodate this setting and reproduce its known results. Moreover, the zigzag construction equips the simple wedge with an additional structure, the preferred Cauchy slice $Z$. We would like to identify the analogous new structure in the original AdS/CFT setting.

This is accomplished as follows: let $B$ be an AdS boundary subregion, i.e., a wedge when viewed as a subset of conformal infinity $\scri$. One can also replace $\scri$ by any globally hyperbolic subset of $\scri$ such as a time-band. The causal wedge $c(B)$ is the double complement of $B$ in the bulk:
\begin{equation}\label{adsred}
    c(B) \equiv (B'_{\bar M} \cap M)'~.
\end{equation}
We may now define 
\begin{equation}
    z(B)\equiv z[c(B)]~.
\end{equation}
See Fig.~\ref{fig:adszigzag} for an illustration. 

Importantly, the causal wedge is always contained within the simple/outermost wedge~\cite{Wall:2012uf}, so the zigzag starts within the (traditionally defined) simple wedge. In other words, despite starting the zigzag already in the bulk, at the edge of $c(B)$, Eq.~\eqref{adsred} does not ``overshoot.'' 

Moreover, the required inclusion of $z[c(B)]$  in $\tilde c(B)'$ ensures that $z$ does not enter $I[c(\bar B)]$, where $\bar B$ is the complement wedge of $B$ on $\scri$. Hence the conformal boundary of $z[c(B)]$ will be $B$, so the usual homology condition is obeyed.

\begin{figure}[b]
    \centering
    \includegraphics[width=16cm]{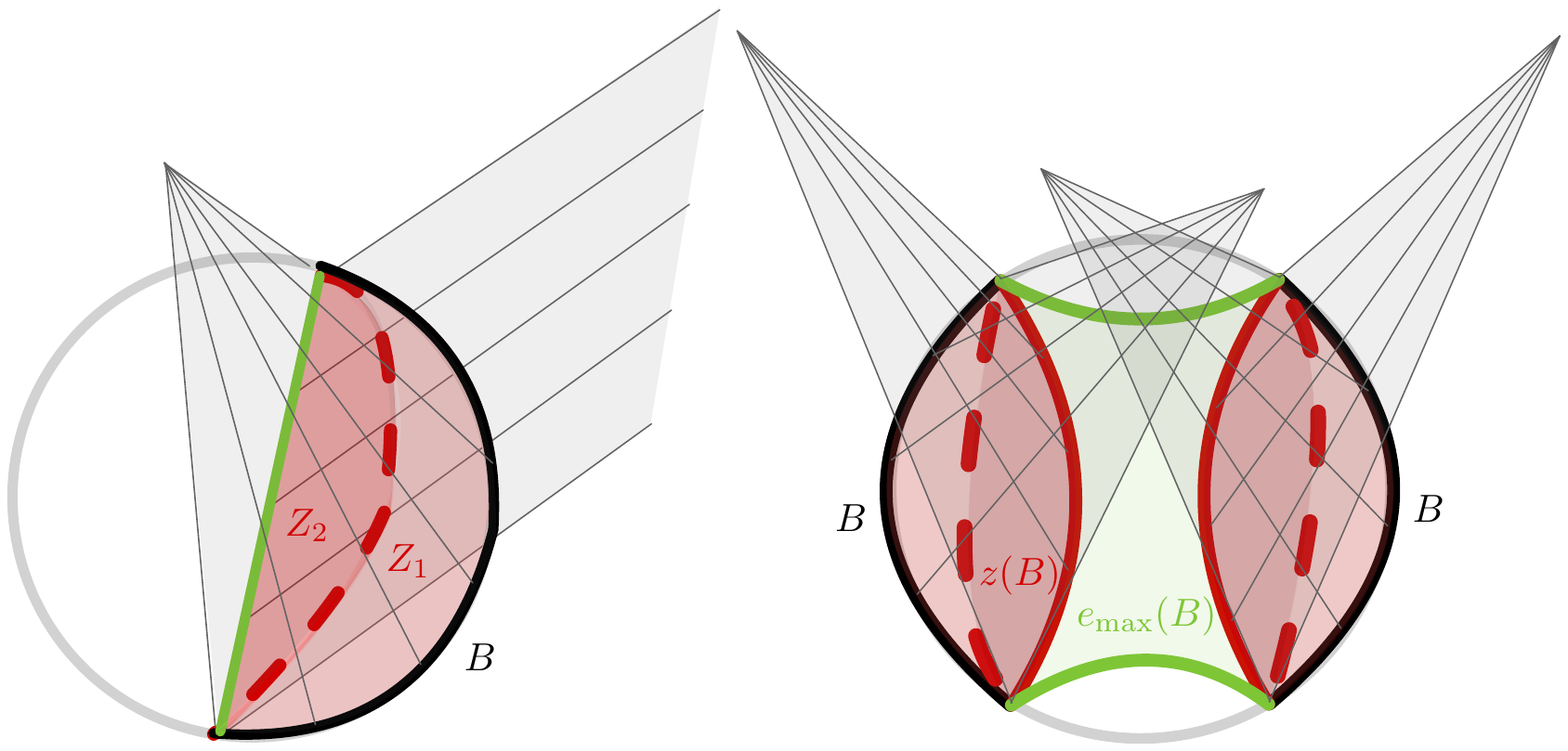}
    \vspace{-4.4cm}
    \caption{Simple wedges in 2+1 dimensional vacuum AdS. In both examples, $B$ is a subset of a very thin time-band on $\scri$,  idealized as a fixed time, so the causal wedge $c(B)$ is not shown. (See the grey wedges in Fig.~\ref{fig:adszigzag} for examples of $c(B)$.) \textit{Left}: $B$ is half of the boundary; $z(B)=\emax(B)$ is half of the bulk. The zig $Z_1$ is a portion of a light cone; $Z_2$ is a portion of the Rindler horizon. \textit{Right}: $B$ consists of two disconnected intervals. The simple wedge (red) is disconnected; the larger entanglement wedge (green) is connected.}
    \label{fig:intervals}
\end{figure}

Figure \ref{fig:intervals} shows two additional examples of the zigzag construction of the simple wedge of AdS boundary regions. 

\appendix
\section{Further examples of FNE, PNE, FNC, PNC, and marginal wedges}\label{app:examples}

In this appendix, we present further examples that illustrate the various types of discrete expansions in commonly considered spacetimes. This complements a first set of examples given in Fig.~\ref{fig:ads}.

\begin{figure}[H]
    \centering
    \vspace{.2cm}
    \includegraphics[width=16cm]{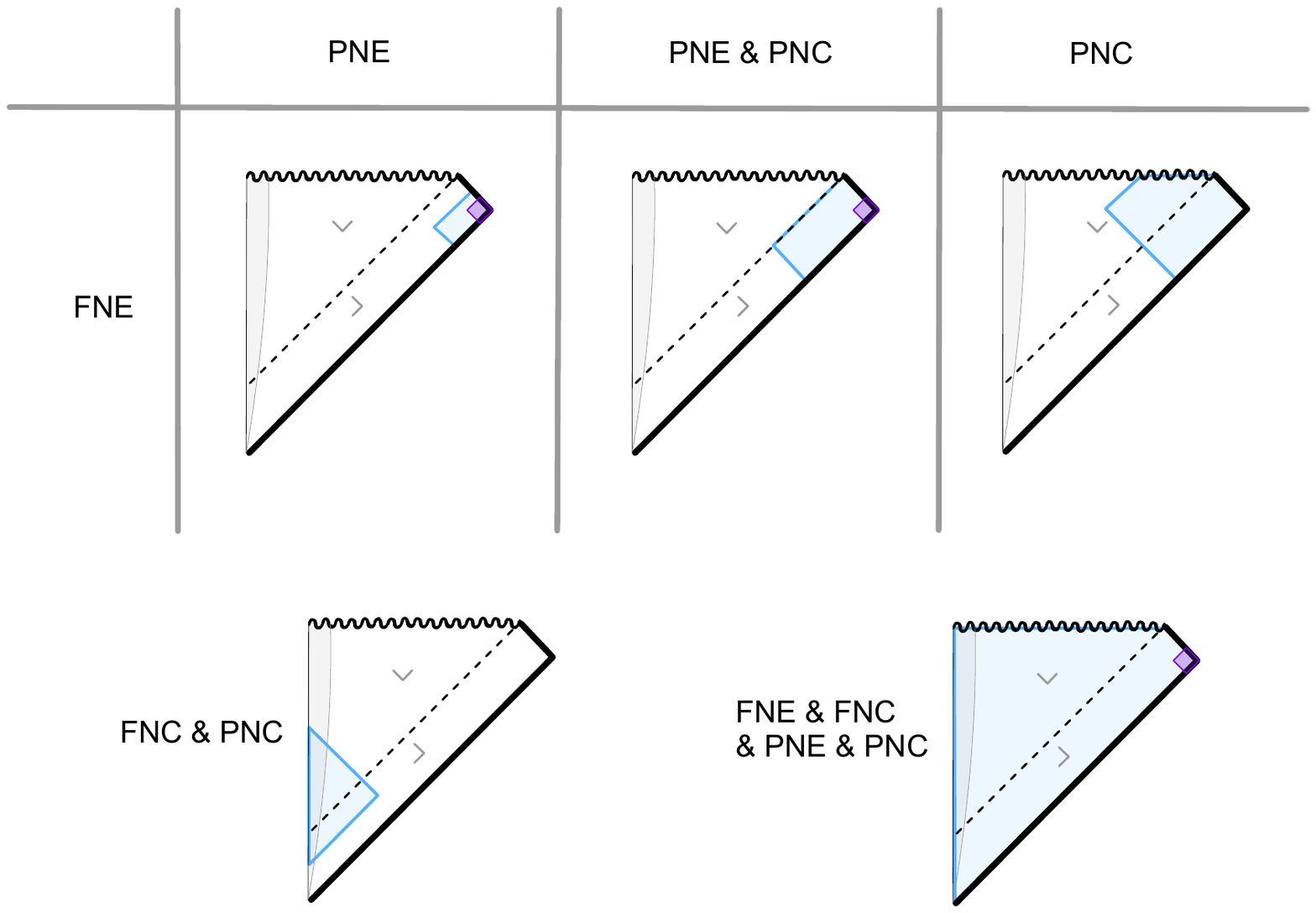}
    \vspace{-.7cm}
    \caption{Classical black hole formed from collapse in asymptotically flat space. Again various examples of wedges are shown, with some being accessible from a purple subwedge. (Past-marginally accessible in the middle top figure; throat accessible in the bottom right wedge.)}
    \label{fig:flat}
\end{figure}

\begin{figure}[H]
    \centering
    \vspace{-1.5cm}
    \includegraphics[width=14.5cm]{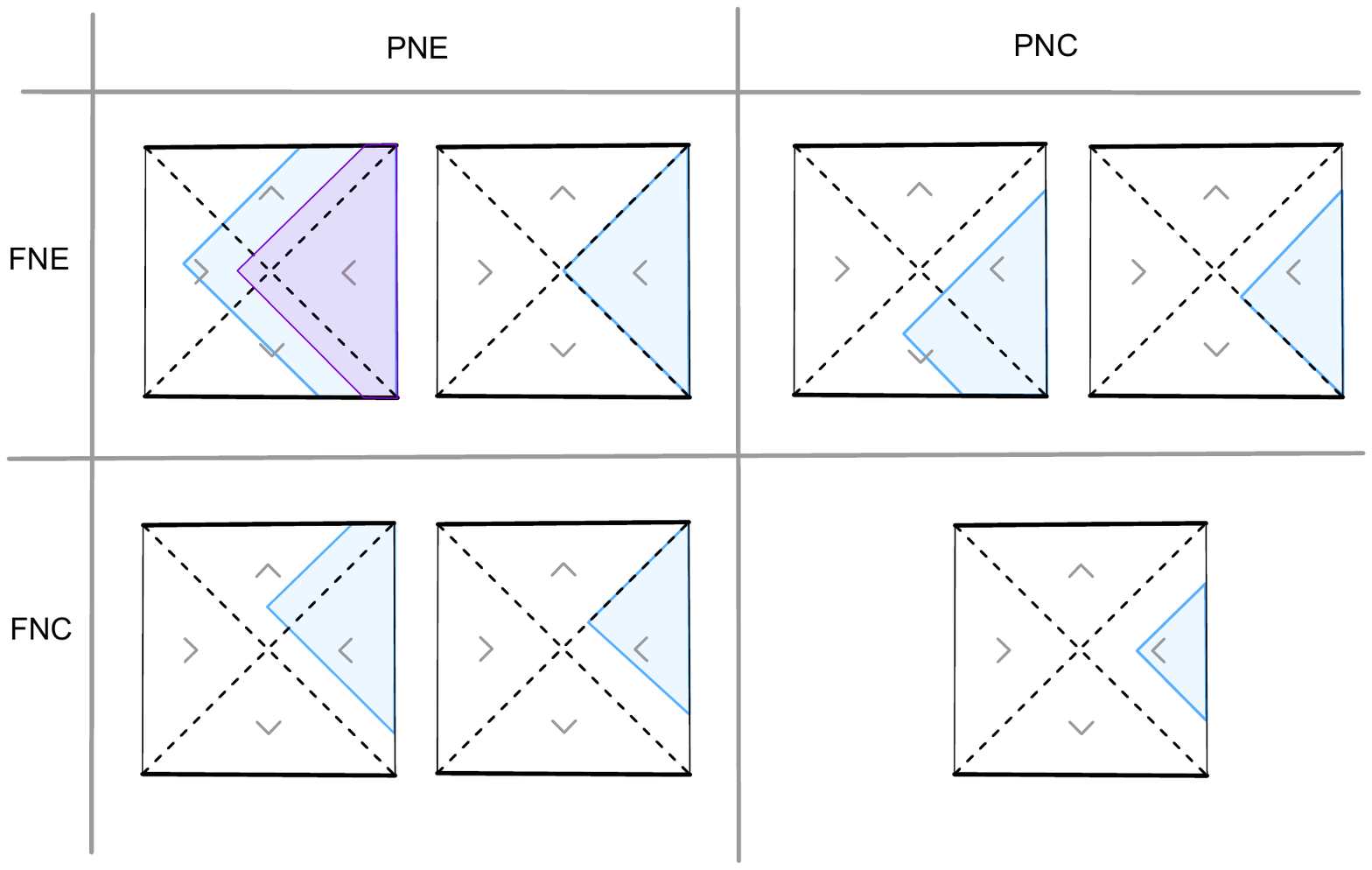}
    \vspace{-1.5cm}
    \caption{Vacuum de Sitter space. Except for the top left example, none of the blue wedges are accessible from any spherical proper subwedge.}
    \label{fig:ds}
\end{figure}

\begin{figure}[H]
    \centering
    \vspace{-.5cm}
    \includegraphics[width=14.5cm]{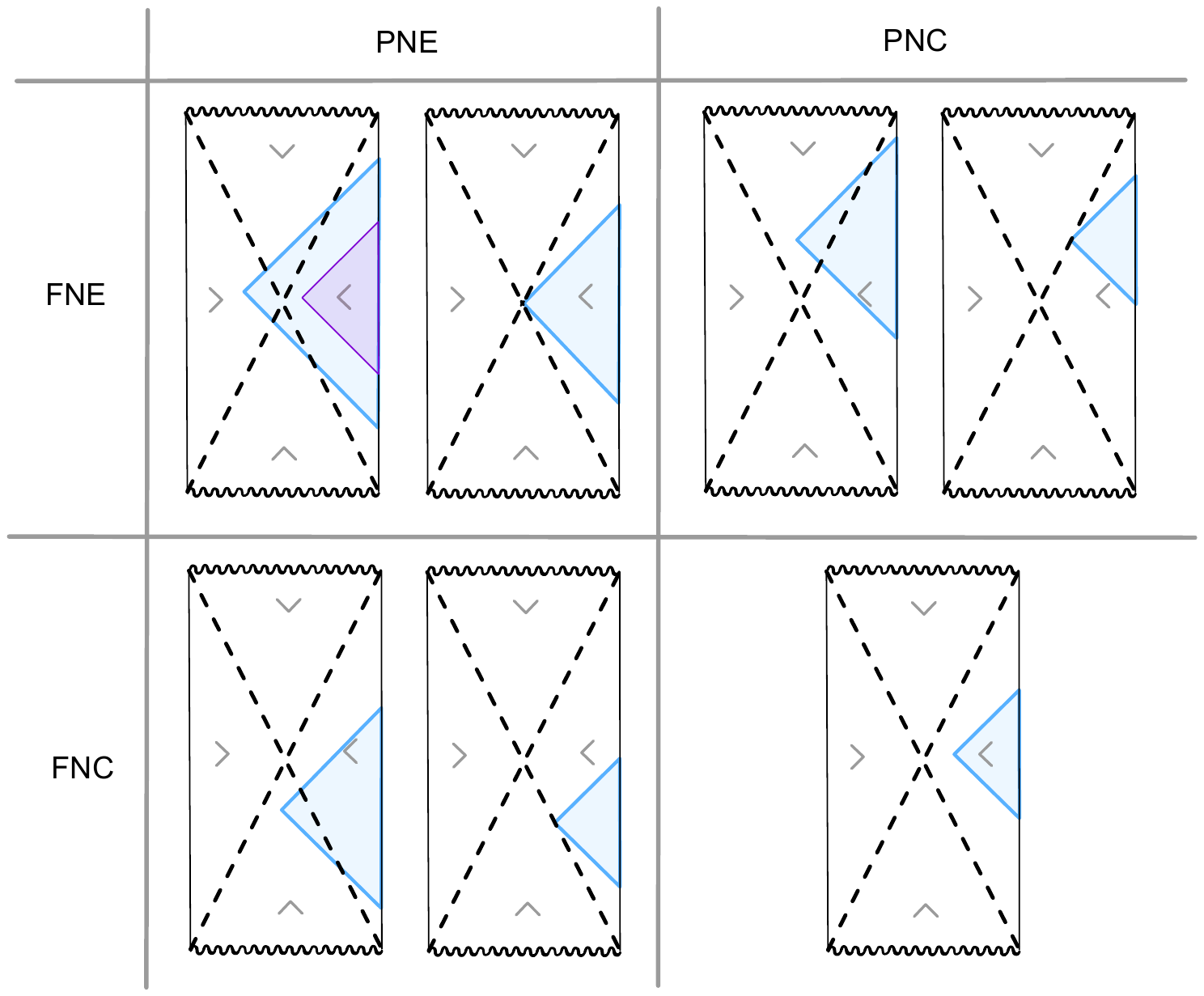}
    \vspace{-.3cm}
    \caption{Spatially closed FLRW universe dominated by pressureless dust. As in vacuum de Sitter, except for the top left example, none of the blue wedges are accessible from any proper subwedge.}
    \label{fig:flrw}
\end{figure}

\subsection*{Acknowledgements}
We would like to thank C.~Akers, N.~Engelhardt, G.~Lin, P.~Rath, A.~Shahbazi-Moghaddam, J.~Yeh, and especially S.~Kaya and G.~Penington for discussions. This work was supported by the Department of Energy, Office of Science, Office of High Energy Physics under QuantISED Award DE-SC0019380.

\bibliographystyle{JHEP}
\bibliography{covariant}

\end{document}